%% file: dmtcs.tex
\documentclass[final]{dmtcs-episciences}

\usepackage[utf8]{inputenc}
\usepackage{subfigure}

\usepackage[ruled]{algorithm2e}
\usepackage{mathtools}
\usepackage{framed}
\usepackage{color,dsfont}
\usepackage{hyphenat}
\usepackage{amssymb,amsmath}
\usepackage{xspace}
\usepackage{flushend}
\usepackage{ntheorem}
\newcommand{\A}{\ensuremath{\mathsf{RSP}}\xspace}

\newcommand{\variable}[1]{#1}
\newcommand{\parent}{\variable{par}}
\newcommand{\enfants}{\variable{children}}
\newcommand{\status}{\variable{st}}
\newcommand{\dist}{\variable{d}}
\newcommand{\abnormalRoot}{abRoot}
\newcommand{\compute}{computePath}
\newcommand{\weight}{w}
\newcommand{\seg}{\mathsf{seg}}
\newcommand{\maxi}{\texttt{W}_{\max}}

\newcommand{\X}{X}
\newcommand{\nmax}{{n_{\texttt{maxCC}}}}

\newtheorem{definition}{Definition}
\newtheorem{lemma}{Lemma}
\newtheorem{theorem}{Theorem}
\newtheorem{corollary}{Corollary}

\DeclareMathOperator*{\argmin}{\arg\!\min}

\renewcommand{\And}{\wedge}

\newcommand{\PReset}{\ensuremath{P\_reset}\xspace}
\newcommand{\PCorrection}{\ensuremath{P\_correction}\xspace}

\newcommand{\Ref}[2]{\hyperref[#2]{\textit{#1 \ref{#2}}}}
\usepackage{graphicx}        % standard LaTeX graphics tool
\usepackage{multicol}        % used for the two-column index
\usepackage[bottom]{footmisc}% places footnotes at page bottom

\newtheorem{observation}{Observation}

\author{St\'ephane Devismes\affiliationmark{1}%\thanks{I am fully supported.}
  \and David Ilcinkas\affiliationmark{2}%\thanks{And he is, too!}
  \and Colette Johnen\affiliationmark{2}}

\title[Self-Stabilizing Disconnected Components Detection and Rooted
  Shortest-Path Tree]{Self-Stabilizing Disconnected Components
  Detection and Rooted Shortest-Path Tree Maintenance in Polynomial
  Steps\thanks{This study has been partially supported by the
    \textsc{anr} projects \textsc{Descartes} (ANR-16-CE40-0023),
    \textsc{Estate} (ANR-16-CE25-0009), and \textsc{Macaron}
    (ANR-13-JS02-002).  This study has been carried out in the frame
    of ``the Investments for the future'' Programme IdEx Bordeaux --
    \textsc{cpu} (ANR-10-IDEX-03-02).  A preliminary version of this
    paper appeared in the Proceedings of the 20th International
    Conference on Principles of Distributed Systems (OPODIS
    2016)~\cite{DIJ16}.}}

\affiliation{
  Universit\'{e} Grenoble Alpes, Grenoble, France\\
Univ. Bordeaux \& CNRS, LaBRI, UMR 5800, F-33400 Talence, France}

\keywords{distributed algorithm, self-stabilization, routing
  algorithm, shortest path, disconnected network, shortest-path tree}

\received{2017-3-9}
\revised{2017-9-22}
\accepted{2017-11-27}

\begin{document}
\publicationdetails{19}{2017}{3}{14}{3181}

\maketitle

\begin{abstract}
 \input{abs}
\end{abstract}

\input{intro}
\input{model}

\input{algo}
\input{proof}
\input{rounds}
\input{ccl}

%\acknowledgements

\nocite{*}
%\bibliographystyle{abbrvnat}
% use the following instead if you encounter problems 
\bibliographystyle{alpha}
\bibliography{biblio}
\label{sec:biblio}

\end{document}

%% file: abs.tex
We deal with the problem of maintaining a shortest-path tree rooted at
some process $r$ in a network that may be disconnected after
topological changes. The goal is then to maintain a shortest-path tree
rooted at $r$ in its connected component, $V_r$, and make all
processes of other components detecting that $r$ is not part of their
connected component.  
We propose, in the composite atomicity model, a silent
self-stabilizing algorithm for this problem working in semi-anonymous
networks, where edges have strictly positive weights.
This algorithm does not require any {\em a priori} knowledge about
global parameters of the network.
We prove its correctness assuming the distributed unfair daemon, the
most general daemon.
Its stabilization time in rounds is at most~$3\nmax+D$, where $\nmax$
is the maximum number of non-root processes in a connected component
and $D$ is the hop-diameter of $V_r$.
Furthermore, if we additionally assume that edge weights are positive
integers, then it stabilizes in a polynomial number of steps: namely,
we exhibit a bound in $O(\maxi \nmax^3 n)$, where $\maxi$ is the
maximum weight of an edge and $n$ is the number of processes.

%% file: intro.tex
\section{Introduction}

Given a connected undirected edge-weighted graph $G$, a {\em
  shortest-path (spanning) tree rooted at node $r$} is a spanning tree
$T$ of $G$, such that for every node $u$, the unique path from $u$ to
$r$ in $T$ is a shortest path from $u$ to $r$ in $G$.  This data
structure finds applications in the networking area ({\em n.b.}, in
this context, nodes actually represent processes), since many
distance-vector routing protocols, like {\em RIP} ({\em Routing
  Information Protocol}) and {\em BGP} ({\em Border Gateway
  Protocol}), are based on the construction of shortest-path
trees. Indeed, such algorithms implicitly builds a shortest-path tree
rooted at each destination.

From time to time, the network may be split into several connected
components due to the network dynamics. In this case, routing to
process $r$ correctly operates only for the processes of its connected
component, $V_r$. Consequently, in other connected components,
information to reach $r$ should be removed to gain space in routing
tables, and to discard messages destined to $r$ (which are unable to
reach $r$ anyway) and thus save bandwidth.
The goal is then to make the network converging to a configuration
where every process of~$V_r$ knows a shortest path to~$r$ and every
other process detects that~$r$ is not in its own connected
component. We call this problem the {\em Disconnected Components
  Detection and rooted Shortest-Path tree Maintenance}
(DCDSPM) problem. Notice that a solution to this problem allows to prevent the
well-known {\em count-to-infinity} problem \cite{LgW04}, where the
distances to some unreachable process keep growing in routing tables
because no process is able to detect the issue.

When topological changes are infrequent, they can be considered as
transient faults~\cite{T01} and self-stabilization~\cite{D74j} --- a
versatile technique to withstand \textit{any} finite number of
transient faults in a distributed system --- becomes an attractive
approach. A self-stabilizing algorithm is able to recover
without external ({\em e.g.}, human) intervention a correct behavior
in finite time, regardless of the \emph{arbitrary} initial
configuration of the system, and therefore, also after the occurrence
of transient faults, provided that these faults do not alter the code
of the processes.

A particular class of self-stabilizing algorithms is that of silent
algorithms. A self-stabilizing algorithm is {\em
  silent}~\cite{DolevGS96} if it converges to a global state where the
values of communication registers used by the algorithm remain fixed.
Silent (self-stabilizing) algorithms are usually proposed to build
distributed data structures, and so are well-suited for the problem
considered here.
As quoted in~\cite{DolevGS96}, the silent property usually implies
more simplicity in the algorithm design, moreover a silent algorithm
may utilize less communication operations and communication bandwidth.

For the sake of simplicity, we consider here a single destination process
$r$, called the {\em root}. However, the solution we will propose can
be generalized to work with any number of destinations, provided that
destinations can be distinguished. In this context, we do not require
the network to be fully identified. Rather, $r$ should be
distinguished among other processes, and all non-root processes are
supposed to be identical: we consider semi-anonymous networks.

In this paper, we propose a silent self-stabilizing algorithm, called
Algorithm \A, for the DCDSPM problem with a single destination process
in semi-anonymous networks. Algorithm \A does not require any {\em a
  priori} knowledge of processes about global parameters of the
network, such as its size or its diameter.  Algorithm \A is written in
the locally shared memory model with composite atomicity introduced by
Dijkstra~\cite{D74j}, which is the most commonly used model in
self-stabilization.  In this model, executions proceed in (atomic)
steps, and a self-stabilizing algorithm is silent if and only if all
its executions are finite.  Moreover, the asynchrony of the system is
captured by the notion of {\em daemon}. The weakest ({\em i.e.}, the
most general) daemon is the {\em distributed unfair daemon}. Hence,
solutions stabilizing under such an assumption are highly desirable,
because they work under any other daemon
assumption. Interestingly, self-stabilizing algorithms designed
  under this assumption are easier to compose (composition techniques
  are widely used to design and prove complex self-stabilizing
  algorithms).  Moreover, time complexity (the {\em stabilization
  time}, mainly) can be bounded in terms of steps only if the
algorithm works under an unfair daemon. Otherwise ({\em e.g.}, under a
weakly fair daemon), time complexity can only be evaluated in terms of
rounds, which capture the execution time according to the slowest
process.  There are many self-stabilizing algorithms proven under the
distributed unfair daemon, {\em
  e.g.},~\cite{ACDDP14,CDDLR15,DLP11,DLV11,GHIJ14}. However, analyses
of the stabilization time in steps is rather unusual and this may be
an important issue. Indeed, this complexity captures the amount
  of computations an algorithm needs to recover a correct behavior.
  Now, recently, several self-stabilizing algorithms, which work
under a distributed unfair daemon, have been shown to have an
exponential stabilization time in steps in the worst
case. In~\cite{ACDDP14}, silent leader election algorithms
from~\cite{DLP11,DLV11} are shown to be exponential in steps in the
worst case. In~\cite{DJ16}, the Breadth-First Search (BFS) algorithm
of Huang and Chen~\cite{HC92} is also shown to be exponential in
steps. Finally, in~\cite{GHIJ16} authors show that the first silent
self-stabilizing algorithm for the DCDSPM problem (still assuming a
single destination) they proposed in~\cite{GHIJ14} is also exponential
in steps.
% Precisely, they exhibit a family of graph of~$4k+2$ nodes on
% which there is an execution of their algorithm containing at
% least~$2^{k+2}$ steps.

\subsection{Contribution} 
Algorithm \A proposed here is proven assuming the distributed unfair daemon. 
We also study its stabilization time in rounds. We establish a
bound of at most~$3\nmax+D$ rounds, where $\nmax$ is the maximum number of
non-root processes in a connected component and $D$ is the hop-diameter of
$V_r$ (defined as the maximum over all pairs $\{u, v\}$ of nodes in
$V_r$ of the minimum number of edges in a shortest path from $u$ to
$v$).

Furthermore, \A is the first silent self-stabilizing algorithm for the
DCDSPM problem which, assuming that the edge weights are positive
integers, achieves a polynomial stabilization time in steps.  Namely,
in this case, the stabilization time of \A is at most
$(\maxi\nmax^3+(3-\maxi)\nmax+3)(n-1)$, where $\maxi$ is the maximum
weight of an edge and $n$ is the number of processes.  ({\em N.b.},
this stabilization time is less than or equal to $\maxi n^4$, for all
$n \geq 3$.)

Finally, notice that when all weights are equal to one, the DCDSPM
problem reduces to a BFS tree maintenance and the step complexity
becomes at most $(\nmax^3 +2\nmax+3)(n-1)$, which is less than or
equal to $n^4$ for all $n \geq 2$.

\subsection{Related Work}

To the best of our knowledge, only one self-stabilizing algorithm for
the DCDSPM problem has been previously proposed in the
literature~\cite{GHIJ14}. This algorithm is silent and works under the
distributed unfair daemon, but, as previously mentioned, it is
exponential in steps. However,
% it assumes positive \emph{real} weights
% whereas Algorithm \A assumes positive \emph{integer}
% weights\footnote{It is not difficult to see that our algorithm is in
%   fact correct even when weights are positive reals. However, our
%   bound on the number of steps is valid only for integer weights, not
%   for arbitrary real ones.}, and 
it has a slightly better
stabilization time in rounds, precisely at most $2(\nmax+1)+D$
rounds\footnote{In fact, \cite{GHIJ14} announced $2n+D$ rounds, but it
  is easy to see that this complexity can be reduced to
  $2(\nmax+1)+D$.}.

There are several shortest-path spanning tree algorithms in the
literature that do not consider the problem of disconnected components
detection.  The oldest distributed algorithms are inspired by the
Bellman-Ford algorithm~\cite{bellman,ford}. Self-stabilizing
shortest-path spanning tree algorithms have then been proposed
in~\cite{CS94,HL02}, but these two algorithms are proven assuming a
central daemon, which only allows sequential executions.  However, in
\cite{huang2005self}, Tetz Huang proves that these algorithms actually
work assuming the distributed unfair daemon. Nevertheless, no upper
bounds on the stabilization time (in rounds or steps) are given. More
recently, Cobb and Huang \cite{CH09} proposed an algorithm constructing
shortest-path trees based on any maximizable routing metrics. This
algorithm does not require a priori knowledge about the network but it
is proven only for the central weakly-fair daemon. It runs in a
linear number of rounds and no analysis is given on the number of steps.

Self-stabilizing shortest-path spanning tree algorithms are also given
in \cite{AGH90,CG02,JT03}.  These algorithms additionally ensure the
loop-free property in the sense that they guarantee that a spanning
tree structure is always preserved while edge costs change
dynamically. However, none of these papers consider the unfair daemon, and
consequently their step complexity cannot be analyzed.
% (although their
% round complexity is given in some cases).

Whenever all edges have weight one, shortest-path trees correspond to
BFS trees. In~\cite{DDL12}, the authors introduce the {\em
  disjunction} problem as follows. Each process has a constant input
bit, 0 or 1. Then, the problem consists for each process in computing
an output value equal to the disjunction of all input bits in the
network. Moreover, each process with input bit 1 (if any) should be
the root of a tree, and each other process should join the tree of the
closest input bit 1 process, if any. If there is no process with input
bit~1, the execution should terminate and all processes should
output~0.  The proposed algorithm is silent and
self-stabilizing. Hence, if we set the input of a process to 1 if and
only if it is the root, then their algorithm solves the DCDSPM problem
when all edge-weights are equal to one, since any process which is not
in $V_r$ will compute an output 0, instead of 1 for the processes in
$V_r$. The authors show that their algorithm stabilizes in $O(n)$
rounds, but no step complexity analysis is given. Now, as their approach
is similar to \cite{DLV11}, it is not difficult to see that their
algorithm is also exponential in steps.

Several other self-stabilizing BFS tree algorithms have been proposed, but
without considering the problem of disconnected components detection.
Chen {\em et al.}  present the first self-stabilizing BFS tree
construction in \cite{CYH91} under the central daemon.  Huang and
Chen present the first self-stabilizing BFS tree construction in
\cite{HC92} under the distributed unfair daemon, but recall that this
algorithm has been proven to be exponential in steps
in~\cite{DJ16}. Finally, notice that these two latter
algorithms~\cite{CYH91,HC92} require that the processes know the exact
number of processes in the network.

According to our knowledge, only the following works
\cite{CDV09,CRV11} take interest in the computation of the number of
steps required by their BFS algorithms.  The algorithm in \cite{CDV09}
is not silent and has a stabilization time in $O(\Delta n^3)$ steps,
where $\Delta$ is the maximum degree in the network.  The silent
algorithm given in \cite{CRV11} has a stabilization time~$O(D^2)$
rounds and~$O(n^6)$ steps. 

Silent self-stabilizing algorithms that construct spanning trees of
arbitrary topologies are
given in~\cite{Cournier09,KK05}.  The solution proposed
in~\cite{Cournier09} stabilizes in at most $4n$ rounds and
$5n^2$ steps, while the algorithm given in~\cite{KK05}
stabilizes in $nD$ steps (its round complexity is not analyzed).

Several other papers propose self-stabilizing algorithms stabilizing
in both a polynomial number of rounds and a polynomial number of
steps, {\em e.g.},~\cite{ACDDP14} (for the leader
election),~\cite{CDPV06,CDV05} (for the DFS token circulation). The
silent leader election algorithm proposed in~\cite{ACDDP14} stabilizes
in at most $3n + D$ rounds and $O(n^3)$ steps. DFS token
circulations given in~\cite{CDPV06,CDV05} execute each wave in $O(n)$
rounds and $O(n^2)$ steps using $O(n\log n)$ space per process
for the former, and $O(n^3)$ rounds and $O(n^3)$ steps using $O(\log
n)$ space per process for the latter.

\subsection{Roadmap} In the next section, we present the computational
model and basic definitions. In Section~\ref{sec:algo}, we describe
Algorithm \A. Its proof of correctness and a complexity analysis in
steps are given in Section~\ref{sec:correctness}, whereas an analysis
of the stabilization time in rounds is proposed in
Section~\ref{sect:rounds}. Finally, we make concluding remarks in
Section~\ref{ccl}.

%% file: model.tex
\section{Preliminaries}

We consider {\em distributed systems} made of $n \geq 1$ interconnected
processes. Each process can directly communicate with a subset of
other processes, called its {\em neighbors}. Communication is assumed
to be bidirectional.  Hence, the topology of the system can be
represented as a simple undirected graph $G = (V,E)$, where $V$ is the
set of processes and $E$ the set of edges, representing communication
links.  Every process $v$ can distinguish its neighbors using a {\em
  local labeling} of a given datatype $Lbl$.  All labels of $v$'s
neighbors are stored into the set $\Gamma(v)$.  Moreover, we assume
that each process $v$ can identify its local label in the set
$\Gamma(u)$ of each neighbor $u$. Such labeling is called {\em
  indirect naming} in the literature~\cite{Sloman:1987}. By an abuse
of notation, we use $v$ to designate both the process $v$ itself, and
its local labels.

Each edge~$\{u,v\}$ has a strictly positive %Integer
{\em weight}, 
denoted~by $\omega(u,v)$. This notion naturally extends to paths: the
weight of a path in $G$ is the sum of its edge weights.  The weighted
distance between the processes~$u$ and~$v$, denoted by~$d(u,v)$, is the
minimum weight of a path from~$u$ to~$v$. Of course, $d(u,v) = \infty$
if and only if $u$ and $v$ belong to two distinct connected components
of $G$.

We use the {\em composite atomicity model of
  computation}~\cite{D74j,D00b} in which the processes communicate
using a finite number of locally shared registers, called {\em
  variables}.  Each process can read its own variables and those of
its neighbors, but can write only to its own variables.  The
\textit{state} of a process is defined by the values of its local
variables.  A {\em configuration} of the system is a vector consisting
of the states of every process. %We denote by~$\mathcal{C}$ the
%set of all possible configurations.

A \emph{distributed algorithm\/} consists of one local program per
process.  We consider semi-uniform algorithms, meaning that all
processes except one, the {\em root} $r$, execute the same program. In
the following, for every process $u$, we denote by $V_u$ the set of
processes (including $u$) in the same connected component of $G$ as
$u$. In the following $V_u$ is simply referred to as the connected
component of $u$. We denote by $\nmax$ the maximum number of non-root
processes in a connected component of $G$. By definition, $\nmax \leq
n-1$.

The \textit{program} of each process consists of a finite set of
\textit{rules} of the form\ $label\ :\ guard \to\ action$.  {\em
  Labels} are only used to identify rules in the reasoning.  A
\textit{guard} is a Boolean predicate involving the state of the
process and that of its neighbors.  The {\em action\/} part of a rule
updates the state of the process.  A rule can be executed only if its
guard evaluates to {\em true}; in this case, the rule is said to be
{\em enabled\/}.  A process is said to be enabled if at least one of
its rules is enabled.  We denote by $\mbox{\it Enabled}(\gamma)$ the
subset of processes that are enabled in configuration~$\gamma$.

 When the configuration is $\gamma$ and
$\mbox{\it Enabled}(\gamma) \neq \emptyset$,
a non-empty set $\mathcal X \subseteq \mbox{\it Enabled}(\gamma)$ is selected; then
every process of $\mathcal X$ {\em atomically} executes one of its
enabled rules, leading to a new configuration $\gamma^\prime$, and so
on.  The transition from $\gamma$ to $\gamma^\prime$ is called a {\em
  step}. The possible steps induce a binary
relation over $\mathcal C$, denoted by $\mapsto$.  An
{\em execution\/} is a maximal sequence of 
configurations $e=\gamma_0\gamma_1\ldots \gamma_i\ldots$ such that
$\gamma_{i-1}\mapsto\gamma_i$ for all $i>0$. The term ``maximal''
means that the execution is either infinite, or ends at a {\em
  terminal\/} configuration in which no rule is
enabled at any process.  

Each step from a configuration to another is driven by a daemon.  We
define a daemon as a predicate over executions.  We said that an
execution $e$ is {\em an execution under the daemon $S$}, if $S(e)$
holds.  In this paper we assume that the daemon is {\em distributed}
and {\em unfair}. ``Distributed'' means that while the configuration
is not terminal, the daemon should select at least one enabled
process, maybe more. ``Unfair'' means that there is no fairness
constraint, {\em i.e.}, the daemon might never select an enabled
process unless it is the only enabled process. In other words, the
distributed unfair daemon corresponds to the predicate $true$, {\em
  i.e.}, this is the most general daemon.

In the composite atomicity model, an algorithm is {\em silent\/} if
 all its possible executions are finite. Hence, we can define silent
self-stabilization as follows.

\begin{definition}
[Silent Self-Stabilization]
\label{def:selfstabilization}
Let~$\mathcal{L}$ be a non-empty subset of configurations, called set
of legitimate configurations.  A distributed system is silent and
self-stabili\-zing under the daemon~$S$ for~$\mathcal{L}$ if and only
if the following two conditions hold:
\begin{itemize}
\item all executions under~$S$ are finite, and
\item all terminal configurations belong to~$\mathcal{L}$.
\end{itemize}
\end{definition}

We use the notion of \textit{round}~\cite{DIM93} to measure the time
complexity.  The definition of round uses the concept of {\em
  neutralization}: a process~$v$ is \textit{neutralized} during a
step~$\gamma_i \mapsto \gamma_{i+1}$, if~$v$ is enabled in~$\gamma_i$
but not in configuration~$\gamma_{i+1}$.  Then, the rounds are
inductively defined as follows. The first round of an execution~$e =
\gamma_0, \gamma_1, \cdots$ is the minimal prefix~$e' = \gamma_0,
\cdots, \gamma_j$, such that every process that is enabled
in~$\gamma_0$ either executes a rule or is neutralized during a step
of~$e'$. Let~$e''$ be the suffix $\gamma_j, \gamma_{j+1}, \cdots$
of~$e$.  The second round of~$e$ is the first round of~$e''$, and so
on.

The {\em stabilization time} of a silent self-stabilizing algorithm is
the maximum time (in steps or rounds) over every execution possible
under the considered daemon $S$ (starting from any initial
configuration) to reach a terminal (legitimate) configuration.

%% file: algo.tex
\section{Algorithm \A}
\label{sec:algo}

This section is devoted to the presentation of our algorithm,
Algorithm \A (which stands for {\em Rooted Shortest-Path}). The 
code of Algorithm \A is given in Algorithm~\ref{alg:A}.

\begin{algorithm}[htpb]
\footnotesize
\begin{center} 
Macro of \A for any process $u$
\end{center}
\begin{tabular}{lll} 
$\enfants(u)$ & $=$ & $\{v
  \in \Gamma(u) \mid \status_u \neq I \wedge \status_v \neq I
  \wedge \parent_v = u \wedge \dist_v \geq \dist_u + \omega(v,u) \wedge
  (\status_v = \status_{u} \vee \status_{u} = EB)\}$\\
\end{tabular}\\

\smallskip

\dotfill

\begin{center} 
Code of \A for the root process $r$
\end{center}
\noindent \textbf{Constants:}\\
\begin{tabular}{lll}
 $\status_r$     & $=$ & $C$\\
 $\parent_r$ & $=$ & $\perp$\\
 $\dist_r$   & $=$ & $0$\\
\end{tabular}\\

\smallskip

\dotfill

\begin{center} 
Code of \A for any process~$u \neq r$
\end{center}

\noindent \textbf{Variables:}\\
\begin{tabular}{lll}
 $\status_u$     & $\in$ & $\{I,C,EB,EF\}$\\
 $\parent_u$ & $\in$ & $Lbl$\\
 $\dist_u$   & $\in$ & $\mathds R^+$
\end{tabular}\\
\noindent \textbf{Predicates:} \\
\begin{tabular}{lll}
$\abnormalRoot(u)$ & $\equiv$ & $\status_u \neq I \wedge \big[ \parent_u
  \notin \Gamma(u) \vee \status_{\parent_u} = I \vee \dist_u <
  \dist_{\parent_u} + \omega(u,\parent_u) \vee$\\
&& $(\status_u \neq
  \status_{\parent_u} \wedge \status_{\parent_u} \neq EB) \big]$\\
$\PReset(u)$ & $\equiv$ & $\status_u = EF \wedge \abnormalRoot(u)$ \\
$\PCorrection(u)$ & $\equiv$ & 
$(\exists v \in \Gamma(u) \mid \status_v = C 
\wedge \dist_v+ \omega(u,v)  < \dist_u)$
\end{tabular}\\
\noindent \textbf{Macro:} \\
\begin{tabular}{lll}
$\compute(u)$ & : & $\parent_u := \argmin_{(v \in \Gamma(u)
 	\And \status_v = C)}(\dist_v + \omega(u,v))$; \\
&&  $\dist_u  := \dist_{\parent_u} + \omega(u, \parent_u)$;\\	 		
&& $\status_u := C$	
\end{tabular}\\
{\textbf{Rules}} \\
\begin{tabular}{lllll}
$\mathbf{R_C}(u)$   & : & $\status_u = C \wedge \PCorrection(u)$                              & $\to$ & $\compute(u)$\\
$\mathbf{R_{EB}}(u)$ & : & $\status_u = C \wedge \neg\PCorrection(u) \wedge$                   & $\to$ & $\status_u := EB$ \\
                    &   & \qquad ($\abnormalRoot(u)$ $\vee$ $\status_{\parent_u}=EB$)           &       &\\
$\mathbf{R_{EF}}(u)$ & : & $\status_u = EB \wedge (\forall v \in \enfants(u) \mid \status_v = EF)$ & $\to$ & $\status_u := EF$\\
$\mathbf{R_I}(u)$   & : & $\PReset(u) \And (\forall v \in \Gamma(u) \mid \status_v \neq C)$   & $\to$ & $\status_u := I$\\
$\mathbf{R_R}(u)$   & : & $(\PReset(u) \vee \status_u = I) \And (\exists v \in \Gamma(u) \mid \status_v = C)$ & $\to$ & $\compute(u)$
\end{tabular}
\caption{Code of \A %for any process~$u \neq r$
\label{alg:A}}
\end{algorithm}

\subsection{Variables}
In \A, each process $u$ maintains three variables: $\status_u$,
$\parent_u$, and $\dist_u$.
Those three variables are constant for the root process\footnote{We
  should emphasize that the use of constants at the root is not a
  limitation, rather it allows to simplify the design and proof of the
  algorithm. Indeed, these constants can be removed by adding a rule
  to correct all root's variables, if necessary, within a single
  step.}, $r$: $\status_r = C$, $\parent_r = \perp$\footnote{$\perp$
  is a particular value which is different from any value in $Lbl$.},
and $\dist_r = 0$.
For each non-root process $u$, we have:
\begin{itemize}
\item $\status_u \in \{I,C,EB,EF\}$, this variable gives the {\em
    status} of the process. $I$, $C$, $EB$, and $EF$ respectively
  stand for {\em Isolated}, {\em Correct}, {\em Error Broadcast}, and
  {\em Error Feedback}. The two first states, $I$ and $C$, are
  involved in the normal behavior of the algorithm, while the two last
  ones, $EB$ and $EF$, are used during the correction
  mechanism. Precisely, $\status_u = C$ (resp. $\status_u = I)$ means
  that $u$ believes it is in $V_r$ (resp. not in $V_r$). The meaning
  of status $EB$ and $EF$ will be further detailed in
  Subsection~\ref{err:corr}.

\item $\parent_u \in Lbl$, a {\em parent pointer}. If $u \in V_r$,
  $\parent_u$ should designate a neighbor of $u$, referred to as its
  {\em parent}, and in a terminal configuration, the parent pointers
  exhibit a shortest path from $u$ to $r$.

  Otherwise ($u \notin V_r$), the variable is meaningless. 
\item $\dist_u \in \mathds{R}^+$, the {\em distance} value. If $u \in V_r$, then
  in a terminal configuration, $\dist_u$ gives the weight of the
  shortest path from $u$ to $r$.

  Otherwise ($u \notin V_r$), the variable is meaningless.
\end{itemize}

\subsection{Normal Execution}

Consider any configuration, where every process $u \neq r$ satisfies
$\status_u = I$, and refer to such a configuration as a {\em normal
  initial configuration}. Each configuration reachable from a {\em
  normal initial configuration} is called a {\em normal
  configuration}, otherwise it is an {\em abnormal
  configuration}. Recall that $\status_r = C$ in all
configurations. Then, starting from a normal initial configuration,
all processes in a connected component different from $V_r$ are
disabled forever. Focus now on the connected component $V_r$. Each
neighbor $u$ of $r$ is enabled to execute $\mathbf{R_R}(u)$. A process
eventually chooses $r$ as parent by executing this rule, which in
particular sets its status to $C$. Then, executions of rule
$\mathbf{R_R}$ are asynchronously propagated in $V_r$ until all its
processes have status $C$: when a process $u$ with status $I$ finds
one of its neighbor with status $C$ it executes $\mathbf{R_R}(u)$, {\em i.e.} $u$
takes status $C$ and chooses as parent its neighbor $v$ with status
$C$ such that $\dist_v + \omega(u,v)$ is minimum, $\dist_u$ being updated
accordingly.  In parallel, rules $\mathbf{R_C}$ are executed to reduce
the weight of the tree rooted at $r$: when a process $u$ with status
$C$ can reduce $\dist_u$ by selecting another neighbor with status $C$
as parent, it chooses the one allowing to minimize $\dist_u$ by
executing $\mathbf{R_C}(u)$. Hence, eventually, the system reaches a
terminal configuration, where the tree rooted at $r$ is a
shortest-path tree spanning all processes of $V_r$.

\subsection{Error Correction}\label{err:corr}

Assume now that the system is in an abnormal
configuration. Thanks to the predicate $\abnormalRoot$, some
non-root processes locally detect that their state is inconsistent
with that of their neighbors. We call {\em abnormal roots} such
processes.  Informally (see Algorithm \A for the formal
  definition), a process~$u\neq r$ is an {\em abnormal root} if~$u$
is not isolated ({\em i.e.}, $\status_u \neq I$) and satisfies one of
the following four conditions:
\begin{enumerate}
\item its parent pointer does not designate a neighbor,
\item its parent has status~$I$,
\item its distance value~$\dist_u$ is inconsistent with 
the distance value of its parent, or
\item its status is inconsistent with 
the status of its parent.
\end{enumerate}
Every non-root process $u$ that is not an abnormal root satisfies one
of the two following cases.  Either $u$ is {\em isolated}, {\em i.e.},
$\status_u = I$, or $u$ points to some neighbor ({\em i.e.},
$\parent_u \in \Gamma(u)$) and the state of $u$ is coherent {\em
  w.r.t.} the state of its parent. In this latter case, $u \in
\enfants(\parent_u)$, {\em i.e.}, $u$ is a ``real'' child of its
parent (see Algorithm \A for the formal definition).  Consider a path
$\mathcal P = u_1, \ldots, u_k$ (with $k \geq 1$) such that $u_1$ is either $r$ or an
abnormal root, and $\forall i, 1 \leq i < k, u_{i+1} \in
\enfants(u_i)$.  $\mathcal P$ is acyclic and called a {\em branch}
rooted at $u_1$.  Hence, we define the normal tree $T(r)$ (resp. an
abnormal tree $T(v)$, for any abnormal root $v$) as the set of all
processes that belong to a branch rooted at $r$ (resp. $v$).

Then, the goal is to remove all abnormal trees so that the system
recovers a normal configuration. For each abnormal tree $T$, we have
two cases. In the former case, the abnormal root $u$ of $T$ can join
another tree $T'$ using rule $\mathbf{R_C}(u)$, making $T$ a
subtree of $T'$.  In the latter case, $T$ is entirely removed in a
top-down manner starting from its (abnormal) root $u$.  Now, in that
case, we have to prevent the following situation: $u$ leaves $T$; this
removal creates some trees, each of those is rooted at a previous
child of $u$; and later $u$ joins one of those (created) trees.
Hence, the idea is to freeze $T$, before removing it.  By freezing we
mean assigning each member of the tree to a particular state, here
$EF$, so that (1) no member $v$ of the tree is allowed to execute
$\mathbf{R_C}(v)$, and (2) no process $w$ can join the tree by
executing $\mathbf{R_R}(w)$.  Once frozen, the tree can be safely
deleted from its root to its leaves.

The freezing mechanism (inspired from~\cite{BlinCV03}) is achieved
using the status $EB$ and $EF$, and the rules $\mathbf{R_{EB}}$
  and $\mathbf{R_{EF}}$.  If a process is not involved into any
freezing operation, then its status is $I$ or $C$. Otherwise, it has
status $EB$ or $EF$ and no neighbor can select it as its parent. These
two latter states are actually used to perform a ``Propagation of
Information with Feedback'' \cite{C82j,S83j} in the abnormal
trees. This is why status $EB$ means ``Error Broadcast'' and $EF$
means ``Error Feedback''.  From an abnormal root, the status $EB$ is
broadcast down in the tree using rule $\mathbf{R_{EB}}$. Then,
once the $EB$ wave reaches a leaf, the leaf initiates a convergecast
$EF$-wave using rule $\mathbf{R_{EF}}$. Once the $EF$-wave
reaches the abnormal root, the tree is said to be {\em dead}, meaning
that all processes in the tree have status $EF$ and, consequently, no
other process can join it. So, the tree can be safely deleted from its
abnormal root toward its leaves. There is two possibilities for the
deletion. If the process $u$ to be deleted has a neighbor with status
$C$, then it executes rule $\mathbf{R_R}(u)$ to directly join another
``alive'' tree. Otherwise, $u$ becomes isolated by executing rule
$\mathbf{R_I}(u)$, and $u$ may join another tree later.

Let $u$ be a process belonging to an abnormal tree of which it is not
the root. Let $v$ be its parent. From the previous explanation, it
follows that during the correction, $(\status_v,\status_u) \in
\{(C,C),$ $(EB,C),$ $(EB,EB),$ $(EB,EF),$ $(EF,EF)\}$ until $v$ resets
by $\mathbf{R_R}(v)$ or $\mathbf{R_I}(v)$. Now, due to the arbitrary
initialization, the status of $u$ and $v$ may not be coherent, in this
case $u$ should also be an abnormal root. Precisely, as formally
defined in Algorithm~\ref{alg:A}, the status of $u$ is incoherent {\em
  w.r.t} the status of its parent $v$ if $\status_u \neq \status_v$
and $\status_v \neq EB$.

Actually, the freezing mechanism ensures that if a process is the root
of an abnormal alive tree, it is in that situation since the initial
configuration (see Lemma~\ref{lem:pseudoRoots}, page
\pageref{lem:pseudoRoots}). The polynomial step complexity mainly
relies on this strong property.

\subsection{Example}

An example of synchronous execution of \A is given in
Figure~\ref{fig:ex}.  We consider the network topology given on the
top left of the figure. The names $v_1, \ldots, v_{10}$ are only given
to ease the explanation (recall that we consider semi-anonymous
networks where only the root $r$ is distinguished). The network
contains eleven processes divided into two connected components. Let $v_i$
be a process. In the synchronous execution described from configuration
a) to configuration m), the color of $v_i$ indicates its status
$\status_{v_i}$, according to the legend on the top of the figure. The
number next to $v_i$ gives its distance value, $\dist_{v_i}$. If there
is an arrow outgoing from $v_i$, this arrow designates the neighbor
$u$ of $v_i$ pointed as parent, {\em i.e.}, $\parent_{v_i} =
u$. Otherwise, this means that $\parent_{v_i} \notin \Gamma(v_i)$.

In the initial configuration a), there are two abnormal roots: $v_2$
and $v_9$, indeed $\parent_{v_2} \notin \Gamma(v_2)$ and
$\parent_{v_9} \notin \Gamma(v_9)$. The status of $v_2$ is already
equal to $EB$ and this value should be broadcast down in its
subtree. In contrast, $v_9$ has status $C$ and, consequently, should
initiate the broadcast of $EB$. Note also that $v_{10}$ can reduce its
distance value by modifying its parent pointer. Hence, in the step
a) $\mapsto$ b), $v_9$ takes status $EB$ (rule $\mathbf{R_{EB}}(v_9)$),
$v_{10}$ selects $v_9$ as parent (rule $\mathbf{R_C}(v_{10})$), and
finally $v_3$ the unique child of $v_2$ takes status $EB$ (rule
$\mathbf{R_{EB}}(v_3)$).

In the step b) $\mapsto$ c), $EB$ is propagated down the two abnormal
trees: $v_5$, $v_7$, $v_8$, and $v_{10}$ execute $\mathbf{R_{EB}}$.
In configuration c), the value $EB$ has reached three leaves: $v_5$,
$v_7$, and $v_{10}$. These processes are then enabled to initiate a
convergecast $EF$-wave. Hence, in the step c) $\mapsto$ d),
$v_5$, $v_7$, and $v_{10}$ execute $\mathbf{R_{EF}}$, while the last
leaf $v_4$ takes status $EB$ ($\mathbf{R_{EB}}(v_4)$).

In configuration d), all children of $v_9$ have status $EF$, so $v_9$
is enabled to take status $EF$ too ($\mathbf{R_{EF}}(v_9)$). In
contrast, $v_3$ should wait until its child $v_8$ takes status $EF$.
Hence, in the step d) $\mapsto$ e), $v_9$ takes status $EF$
($\mathbf{R_{EF}}(v_9)$), its abnormal tree becomes frozen, while the
last leaf $v_4$ of the second abnormal tree initiates a convergecast
$EF$-wave (rule $\mathbf{R_{EF}}(v_4)$).

In the step e) $\mapsto$ f), $v_9$ leaves its tree and becomes
isolated by rule $\mathbf{R_I}(v_9)$, while $v_8$ takes status $EF$ by
$\mathbf{R_{EF}}(v_8)$.
Since all its children have now status $EF$, $v_3$ can take status $EF$
by $\mathbf{R_{EF}}(v_3)$ in step f) $\mapsto$ g), while $v5$ and $v_{10}$
become isolated by rule $\mathbf{R_I}$ in the same step.
Remark then that in g), all processes in the connected component $\{v_5,
v_9,v_{10}\}$ are isolated and, since $r$ is not part of this
component, they are disabled forever. In the step g) $\mapsto$ h), the
abnormal root $v_2$ of the remaining abnormal tree takes status $EF$
($\mathbf{R_{EF}}(v_2)$). So, the abnormal tree rooted at $v_2$ is
frozen in configuration h).  In the step h) $\mapsto$ i), $v_2$ leaves
its tree and becomes isolated by rule $\mathbf{R_I}(v_2)$. Then, $v_3$
becomes isolated in step i) $\mapsto$ j) (rule
$\mathbf{R_I}(v_3)$). In step j) $\mapsto$ k), $v_8$ becomes isolated
(rule $\mathbf{R_I}(v_8)$), while $v_7$ joins the normal tree (the
tree rooted at $r$) by rule $\mathbf{R_R}(v_7)$. In the last two
steps, $v_2$, $v_3$, $v_8$, and then $v_4$ successively join the
normal tree by rule $\mathbf{R_R}$, and configuration m) is terminal.

\begin{figure}
\centering
\includegraphics[scale=0.72]{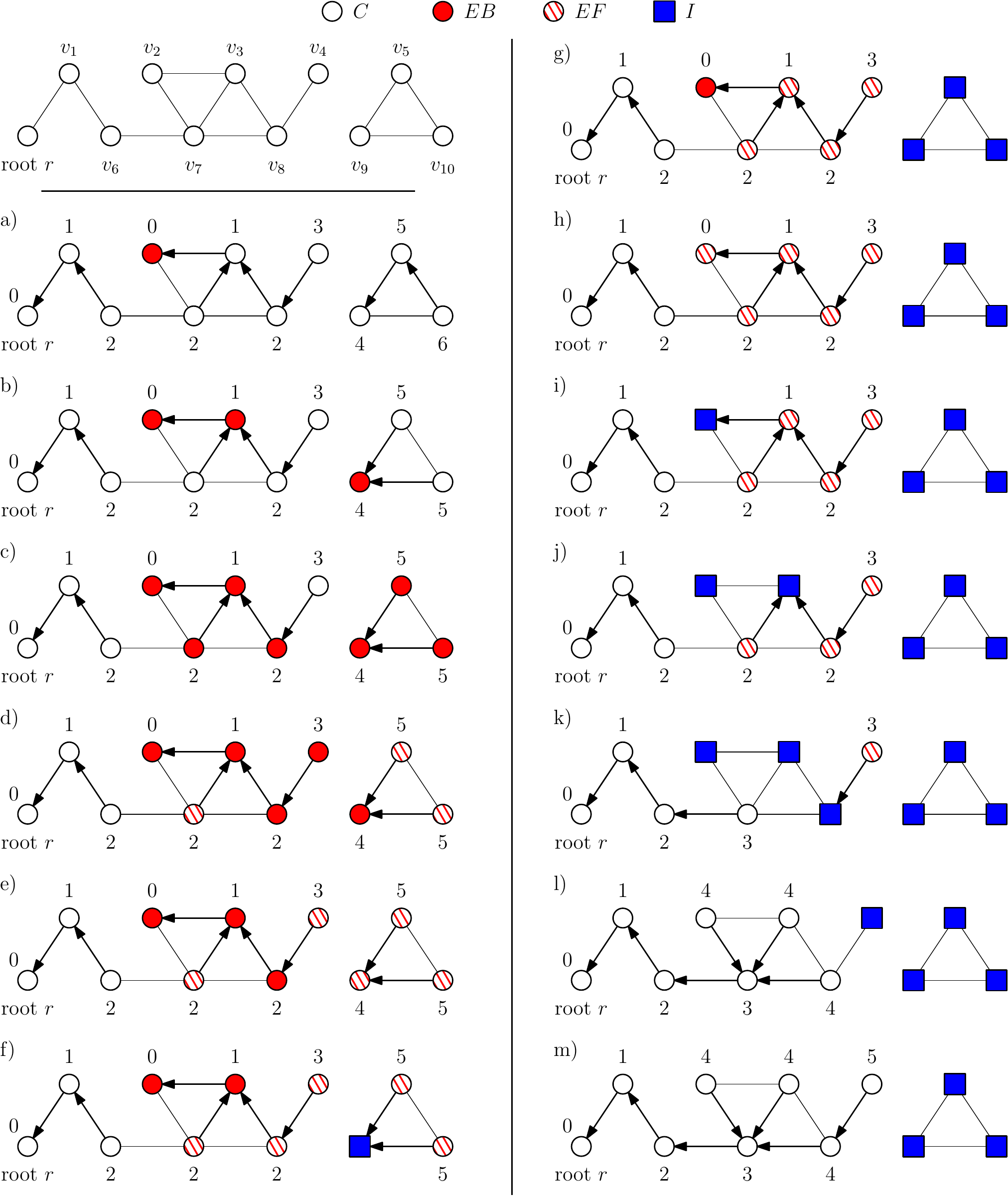}
\caption{A synchronous execution of \A\label{fig:ex}}
\end{figure}

%% file: proof.tex
\section{Correctness and Step Complexity  of Algorithm \A}
\label{sec:correctness}

\subsection{Definitions}\label{pred}

Before proceeding with the proof of correctness and the step
complexity analysis, we define some useful concepts.

\begin{definition}[Abnormal Root] 
  % For every process $u \neq r$, $\abnormalRoot(u) \equiv \status_u
  % \neq I \wedge \big[ \parent_u \notin \Gamma(u) \vee
  % \status_{\parent_u} = I \vee \dist_u < \dist_{\parent_u} +
  % \omega(u,\parent_u) \vee (\status_u \neq \status_{\parent_u}
  % \wedge \status_{\parent_u} \neq EB) \big]$.

  Every process $u \neq r$ that satisfies $\abnormalRoot(u)$ is said to
  be an {\em abnormal root}.
\end{definition}

\begin{definition}[Alive Abnormal Root]
  A process~$u \neq r$ is said to be an \emph{alive abnormal root}
  (resp. a \emph{dead abnormal root}) if~$u$ is an abnormal root and
  has a status different from~$EF$ (resp. has status $EF$).
\end{definition}

%\subsubsection{Children} 

%Informally, for every process~$u \neq r$, either $u$ is {\em isolated}
%({\em i.e.} it has the status~$I$), or $u$ is an {\em abnormal root},
%or $u$ is a child of its parent, {\em i.e.}, a member of the
%set~$\enfants(\parent_u)$, as defined below.

% \begin{definition}[Children] For every process $v$, $\enfants(v) = \{u
%   \in \Gamma(v) \mid \status_v \neq I \wedge \status_u \neq I
%   \wedge \parent_u = v \wedge \dist_u \geq \dist_v + \omega(u,v) \wedge
%   (\status_u = \status_{v} \vee \status_{v} = EB)\}$.
% \end{definition}

\begin{definition}[Branch]
  A \emph{branch} is a 
  % maximal
  sequence of processes~$v_1, \cdots ,v_k$ for some integer~$k \geq
  1$, such that~$v_1$ is~$r$ or an abnormal root and, for every~$1
  \leq i < k$, we have~$v_{i+1} \in \enfants(v_{i})$.
  The process~$v_i$ is said to be at \emph{depth}~$i$ and $v_i, \cdots
  ,v_k$ is called a {\em sub-branch}.
  If~$v_1 \neq r$, the branch is said to be \emph{illegal}, otherwise,
  the branch is said to be \emph{legal}.
\end{definition}

\begin{observation}\label{obs:size}
  A branch depth is at most~$\nmax$.  
  A process~$v$ having status~$I$ does not belong to any branch.  
  If a process~$v$ has status~$C$ (resp.~$EF$), then all processes of a
  sub-branch starting at~$v$ have status~$C$ (resp.~$EF$).
\end{observation}

\begin{definition}[Legitimate State]\label{def:leg:state}
  A process~$u$ is said to be in a \emph{legitimate state} if $u$
  satisfies one of the following three conditions:
  \begin{enumerate}
  \item $u = r$,
  \item\label{def:2} $u \neq r$, $u \in V_r$, $\status_u = C$,
    $\dist_u = d(u,r)$, and $\dist_u =
    \dist_{\parent_u}+\omega(u,\parent_u)$, or
  \item $u \notin V_r$ and $\status_u = I$.
  \end{enumerate}
\end{definition}

\begin{observation}
\label{obs:0}
Every process $u \neq r$ such that $\status_u = C$ and $\dist_u \neq
\dist_{\parent_u}+\omega(u,\parent_u)$ is enabled.
\end{observation}

\begin{definition}[Legitimate Configuration]\label{def:leg}
  A {\em legitimate configuration} is any configuration where every
  process is in a legitimate state.  We denote by $\mathcal{LC}_{\A}$
  the set of all legitimate configurations of Algorithm \A.
\end{definition}

Let $\gamma$ be a configuration. Let $T_\gamma = (V_r,E_{T_\gamma})$
be the subgraph, where $E_{T_\gamma} = \{\{p,q\} \in E\ |\ p \in V_r
\setminus \{r\} \wedge \parent_p = q\}$. By Definition~\ref{def:leg:state} (point~\ref{def:2}), we deduce the following
observation.

\begin{observation}
  In every legitimate configuration $\gamma$, $T_\gamma$ is a
  shortest-path tree spanning all processes of $V_r$.
\end{observation}

\subsection{Partial Correctness}

We now prove that the set of terminal configurations is exactly the
set of legitimate configurations. We start by proving the following
intermediate statement.

\begin{lemma}\label{lem:IorC}
In any terminal configuration, every process has either status~$I$ or~$C$.
\end{lemma}
\begin{proof}
  This is trivially true for the root process,~$r$.  Assume that there
  exists a non-root process with status~$EB$ in a terminal
  configuration~$\gamma$.  Consider the non-root process~$u$ with
  status~$EB$ having the largest distance value~$\dist_u$ in~$\gamma$.
  In~$\gamma$, no process~$v$ with status~$C$ can be a child
  of $u$, otherwise either $\mathrm{R_{EB}}$ or~$\mathrm{R_{C}}$ is
  enabled at $v$ in~$\gamma$, a contradiction. Moreover, by maximality
  of $\dist_u$,~$u$ cannot have a child with status~$EB$ in
  $\gamma$. Therefore, in $\gamma$ process~$u$ has no child or it has
  only children with status~$EF$, and thus rule~$\mathrm{R_{EF}}$ is
  enabled at $u$, a contradiction.  Thus, every process has
  status~$C$,~$I$, or~$EF$ in $\gamma$.

  Assume now that there exists a non-root process with status~$EF$ in a
  terminal configuration~$\gamma$.  Consider the process $u$ with
  status~$EF$ having the smallest distance value~$\dist_u$ in~$\gamma$.
  By construction, $u$ is an abnormal root in $\gamma$.
  So, either~$\mathrm{R_{I}}$ or~$\mathrm{R_R}$ is enabled at $u$ in
  $\gamma$, a contradiction.
\end{proof}

The next lemma, Lemma~\ref{lem:leg1}, deals with the connected
components that do not contain $r$, if any. Then, Lemma~\ref{lem:leg2}
deals with the connected component~$V_r$.

\begin{lemma}\label{lem:leg1}
  In any terminal configuration, every process that does not belong
  to~$V_r$ is in a legitimate state.
\end{lemma}
\begin{proof}
  Consider, by contradiction, that there exists a process $u$ that
  belongs to the connected component $CC$ other than $V_r$ which is
  not in a legitimate state in some terminal configuration~$\gamma$.
  By definition, $u$ is not the root, moreover it has status
  $C$ in $\gamma$, by Lemma~\ref{lem:IorC}.
  So, consider the process $v$ of $CC$ with status~$C$ having
    the smallest distance value~$\dist_v$ in $\gamma$.  By
    construction,~$v$ is an abnormal root in $\gamma$. Thus,
    rule~$\mathrm{R_{EB}}$ is enabled at $v$ in $\gamma$, a
    contradiction.
\end{proof}

\begin{lemma}\label{lem:leg2}
  In any terminal configuration, every process of~$V_r$ is in a
  legitimate state.
\end{lemma}
\begin{proof}
  Assume, by contradiction, that there exists a terminal
  configuration $\gamma$ where at least one process in the connected
  component~$V_r$ is not in a legitimate state.

  Assume also that there exists some process of~$V_r$ that has
  status~$I$ in $\gamma$.  Consider now a process~$u$ of~$V_r$ such
  that in $\gamma$, $u$ has status~$I$ and at least one of its
  neighbors has status~$C$. Such a process exists because no process
  has status~$EB$ or~$EF$ in $\gamma$ (Lemma~\ref{lem:IorC}), but at
  least one process of $V_r$ has status~$C$, namely~$r$. Then,
  $\mathrm{R_R}$ is enabled at~$u$ in $\gamma$, a contradiction.  So,
  every process in~$V_r$ must have status~$C$ in $\gamma$.  Moreover,
  for all processes in~$V_r$, we have~$\dist_u = \dist_{\parent_u} +
  \omega(\parent_u,u)$ in $\gamma$, otherwise $\mathrm{R_{C}}$ is
  enabled at some process of $V_r$ in $\gamma$.

  Assume now that there exists a process~$u$ such that~$\dist_u <
  d(u,r)$ in $\gamma$.  Consider a process~$u$ of~$V_r$ having the
  smallest distance value $\dist_u$ among the processes in~$V_r$ such
  that~$\dist_u < d(u,r)$ in $\gamma$.  By definition, $u \neq r$ and
  we have~$\dist_u > \dist_{\parent_u}$ in $\gamma$, so
  $\dist_{\parent_u} \geq d(\parent_u,r)$ in $\gamma$. Hence, we can
  conclude that $\dist_u \geq d(u,r)$ in $\gamma$, a contradiction.
  So, every process~$u$ in~$V_r$ satisfies~$\dist_u \geq d(u,r)$ in
  $\gamma$.

  Finally, assume that there exists a process~$u$ such that~$\dist_u >
  d(u,r)$ in $\gamma$.  Consider a process~$u$ in~$V_r$ having the
  smallest distance to~$r$ among the processes in~$V_r$ such
  that~$\dist_u > d(u,r)$ in $\gamma$.  By definition, $u \neq r$ and
  there exists some process~$v$ in~$\Gamma(u)$ such that $d(u,r) =
  d(v,r) + \omega(u,v)$ in $\gamma$. Thus, we have~$\dist_v = d(v,r)$
  in $\gamma$.  So,~$\mathrm{R_C}$ is enabled at $u$ in $\gamma$, a
  contradiction.
\end{proof}

After noticing that any legitimate configuration is a terminal one (by
construction of the algorithm), we deduce the following corollary from
the two previous lemmas.

\begin{corollary}\label{coro:legterm}
  For every configuration $\gamma$, $\gamma$ is terminal if and only
  if $\gamma$ is legitimate.
\end{corollary}

\subsection{Termination}

In this section, we establish that every execution of Algorithm~$\A$
under a distributed unfair daemon is finite.  Furthermore, we
compute the following bound on the number of steps of every
execution:~$[\maxi\nmax^3+(3-\maxi)\nmax+3](n-1)$, where $n$ is the
number of processes, $\maxi$ is the maximum weight of an edge, and
$\nmax$ is the maximum number of non-root processes in a connected
component, when all weights are strictly positive integers.

\begin{lemma}\label{lem:pseudoRoots}
No alive abnormal root is created along any execution.
\end{lemma}
\begin{proof}
  Let $\gamma \mapsto \gamma'$ be a step. Let~$u$ be a non-root process
  that is not an \emph{alive abnormal root} in~$\gamma$, and let~$v$
  be the process such that $\parent_{u}=v$ in~$\gamma'$.
  If the status of~$u$ is~$EF$ or~$I$ in~$\gamma'$, then~$u$ is not an
  alive abnormal root in~$\gamma'$. So, let us assume now that the
  status of~$u$ is either~$EB$ or~$C$ in~$\gamma'$.

  Consider then the case where~$u$ has status~$EB$ in~$\gamma'$. The
  only rule $u$ can execute in~$\gamma \mapsto \gamma'$
  is~$\mathrm{R_{EB}}$. So, $\status_{u} \in \{C,EB\}$
  in~$\gamma$. Moreover, whether $u$ executes $\mathrm{R_{EB}}$ or
  not, $\parent_u = v$ in~$\gamma$. Since $\status_{u} \in \{C,EB\}$
  and $u$ is not an alive abnormal root in~$\gamma$, we can deduce
  that $u$ is not an abnormal root in $\gamma$ (whether dead or alive). So, if $\status_{u} =
  EB$ in~$\gamma$, then $\status_{v} = EB$ in~$\gamma$ too. Otherwise,
  $u$ has status $C$ in $\gamma$ while not being an abnormal root in
  $\gamma$: it executes $\mathrm{R_{EB}}(u)$ in~$\gamma \mapsto
  \gamma'$ because $\status_{v} = EB$ in~$\gamma$. Hence, in either
  case $v$ has status~$EB$ in~$\gamma$, and this in particular means
  that $v \neq r$ (this status does not exist for $r$).  Moreover, $u$
  belongs to $\enfants(v)$ in~$\gamma$ (again because $\parent_u = v$
  and $u$ is not an abnormal root in~$\gamma$).  So, $v$ is not
  enabled in~$\gamma$ and $u \in \enfants(v)$ remains true in
  $\gamma'$. Hence, we can conclude that~$u$ is still not an alive
  abnormal root in~$\gamma'$.

  Consider now the other case, {\em i.e.}, $u$ has status~$C$
  in~$\gamma'$.  During~$\gamma \mapsto \gamma'$, the only rules
  that~$u$ may execute are~$\mathrm{R_R}$ or~$\mathrm{R_C}$. 
  If $u$ executes $\mathrm{R_R}$ or~$\mathrm{R_C}$, we have
  $\status_{v}=C$ in~$\gamma$ (because it is a requirement to execute
  any of these rules) and consequently, the only rules that~$v$ may
  execute in~$\gamma \mapsto \gamma'$ are~$\mathrm{R_C}$
  or~$\mathrm{R_{EB}}$.
  Otherwise ({\em i.e.}, $u$ does not execute any rule in $\gamma
  \mapsto \gamma'$), $\parent_{u} = v$ and $\status_{u}=C$ already
  hold in~$\gamma$. In this case, $u$ being not an alive abnormal root
  and $\status_{u}=C$ in~$\gamma$ implies that $u \in \enfants(v)$ and
  thus $\status_{v} \in \{C,EB\}$ in $\gamma$, which further implies
  that the only rules that~$v$ may execute in~$\gamma \mapsto \gamma'$
  in this case are~$\mathrm{R_C}$ or~$\mathrm{R_{EB}}$.
  Thus, in either case, during~$\gamma \mapsto \gamma'$,~$v$ either
  takes the status~$EB$, decreases its distance value, or does not
  change the value of its variables.  Consequently, $u$ belongs
  to $\enfants(v)$ in~$\gamma'$, which prevents $u$ from being an
  alive abnormal root in~$\gamma'$.
\end{proof}

Let $AAR(\gamma)$ be the set of alive abnormal roots in any
configuration $\gamma$.  From the previous lemma, we know that, for
every step $\gamma \mapsto \gamma'$, we have $AAR(\gamma') \subseteq
AAR(\gamma)$ (precisely,  for every process $u$ and every step $\gamma \mapsto
\gamma'$, $u \notin AAR(\gamma) \Rightarrow u \notin
AAR(\gamma')$). So, we can use the notion of \emph{$u$-segment}
(inspired from~\cite{ACDDP14}) to bound the total number of steps in
an execution.

\begin{definition}[$u$-Segment]
  Let $u$ be any non-root process.  Let $e = \gamma_0, \gamma_1,
  \cdots$ be an execution.

If there is no step $\gamma_i \mapsto \gamma_{i+1}$ in $e$, where
there is a non-root process in $V_u$ which is an alive abnormal root in
$\gamma_i$, but not in $\gamma_{i+1}$, then the {\em first
  $u$-segment} of $e$ is $e$ itself and there is no other $u$-segment.

Otherwise, let $\gamma_i \mapsto \gamma_{i+1}$ be the first step of
$e$, where there is a non-root process in $V_u$ which is an alive
abnormal root in $\gamma_i$, but not in $\gamma_{i+1}$.  The {\em
  first $u$-segment} of $e$ is the prefix $\gamma_0, \cdots,
\gamma_{i+1}$. The {\em second $u$-segment} of~$e$ is the first
$u$-segment of the suffix~$\gamma_{i+1}, \gamma_{i+2}, \cdots$, and so
forth.
\end{definition}

By Lemma~\ref{lem:pseudoRoots}, we have 

\begin{observation}\label{obs:seg}
  For every non-root process $u$, for every execution $e$, $e$ contains at
  most~$\nmax+1$ $u$-segments, because there are initially at
  most~$\nmax$ alive abnormal roots in $V_u$.
\end{observation}

\begin{lemma}
  Let $u$ be any non-root process. During a $u$-segment, if $u$
  executes the rule~$\mathrm{R_{EF}}$, then $u$ does not execute any
  other rule in the remaining of the $u$-segment.
\end{lemma}
\begin{proof}
  Let $\seg_u$ be a $u$-segment. Let~$s_1$ be a step of $\seg_u$ in
  which $u$ executes~$\mathrm{R_{EF}}$. Let~$s_2$ be the next step in
  which $u$ executes its next rule. (If $s_1$ or $s_2$ do not exist,
  then the lemma trivially holds for $\seg_u$.)
  Just before~$s_1$, all branches containing~$u$ have an alive
  abnormal root, namely the non-root process~$v$ at depth~$1$ in any
  of these branches. (Note that we may have $v=u$.) On the other hand,
  just before~$s_2$, $u$ is the dead abnormal root of all branches it
  belongs to. This implies that $v$ must have executed the
  rule~$\mathrm{R_{EF}}$ in the meantime and thus is not an alive
  abnormal root anymore when the step $s_2$ is executed.  Therefore,
  $s_1$ and~$s_2$ belong to two distinct $u$-segments of the
  execution.
\end{proof}

\begin{corollary}\label{cor:language}
  Let $u$ be a non-root process.  The sequence of rules executed by
  $u$ during a $u$-segment belongs to the following language:
  $(\mathrm{R_I} + \varepsilon)(\mathrm{R_R} +
  \varepsilon)\mathrm{R_C}^*(\mathrm{R_{EB}} +
  \varepsilon)(\mathrm{R_{EF}} + \varepsilon)$.
\end{corollary}

We use the notion of \emph{maximal causal chain} to further analyze
the number of steps in a $u$-segment.

\begin{definition}[Maximal Causal Chain]
  Let $u$ be a non-root process and~$\seg_u$ be any $u$-segment. A
  {\em maximal causal chain of $\seg_u$} rooted at~$u_0 \in V_u$ is a
  maximal sequence of actions $a_1,a_{2},\cdots,a_k$ executed
  in~$\seg_u$ such that the action~$a_1$ sets $\parent_{u_1}$ to~$u_0
  \in V_u$ not later than any other action by~$u_0$ in~$\seg_u$, and
  for all~$2 \leq i \leq k$, the action~$a_i$ sets $\parent_{u_i}$
  to~$u_{i-1}$ after the action~$a_{i-1}$ but not later than
  $u_{i-1}$'s next action.
\end{definition}

\begin{observation}\label{obs:5} ~
\begin{itemize}
\item An action $a_i$ belongs to a maximal causal chain if and only if
  $a_i$ consists in a call to the macro $\compute$ by a non-root
  process.
%\item Any~$\compute$ execution performed by a non-root process~$u$
%  belongs to a maximal causal chain. % executed in a segment.
\item Only actions of Rules~$\mathrm{R_R}$ and~$\mathrm{R_C}$ contain
  the execution of~$\compute$.
\end{itemize}
 Let $u$ be a non-root process and~$\seg_u$ be any $u$-segment. Let~ $a_1,a_{2},\cdots,a_k$ be a maximal
causal chain of $\seg_u$ rooted at~$u_0$.
\begin{itemize}
\item For all~$ 1 \leq i \leq k$, $a_i$ consists in the execution of
  $\compute$ by~$u_i$ ({\em i.e.},~$u_i$ executes the
  rule~$\mathrm{R_R}$ or~$\mathrm{R_C}$) where $u_i \in V_u$.
\item Denote by $ds_{\seg_u,v}$ the distance value of process~$v$ at
  the beginning of~$\seg_u$.  For all~$1 \leq i \leq k$, $a_i$
  sets~$\dist_{u_i}$ to
  $ds_{\seg_u,u_0}+\sum_{j=1}^{j=i}\weight(u_j,u_{j-1})$, where $u_i$
  is the process that executes $a_i$.
\end{itemize}
\end{observation}

For the next lemmas and theorems, we recall that $\nmax \leq n-1$ is
the maximum number of non-root processes in a connected component of
$G$.

\begin{lemma}\label{lem:causalchain}
 % The length of any maximal causal chain is at most~$\nmax$.
  Let $u$ be a non-root process. All actions in a maximal causal chain
  of a $u$-segment are caused by different non-root processes of
  $V_u$.  Moreover, an execution of~$\compute$ by some non-root
  process~$v$ never belongs to any maximal causal chain rooted at~$v$.
\end{lemma}
\begin{proof}
  First note that any rule~$\mathrm{R_C}$ executed by a process~$v$
  makes the value of~$\dist_v$ decrease.

  Assume now, by the contradiction, that there exists a process $v$
  such that, in some maximal causal chain $a_1,a_{2},\cdots,a_k$
  of a $u$-segment, $v$ is used as parent in some action $a_i$
  and executes the action $a_j$, with $j>i$. The value of~$\dist_{v}$
  is strictly larger just after the action~$a_j$ than just before the
  action~$a_i$. This implies that process $v$ must have executed the
  rule~$\mathrm{R_R}$ in the meantime.  So, $a_i$ and $a_j$ are
  executed in two different $u$-segments by Corollary~\ref{cor:language}
  and the fact that $v$ has status $C$ just before the
  action~$a_i$. Consequently, they do not belong to the same maximal
  causal chain, a contradiction.
 
  Therefore, all actions in a maximal causal chain are caused by
  different processes, and a process never executes an action in a
  maximal causal chain it is the root of. As all actions in a maximal
  causal chain are executed by processes in the same connected
  component, we are done.
\end{proof}

\begin{definition}[$S_{\seg_u,v}$]
  Given a non-root process $u$ and a $u$-segment~$\seg_u$, we define
  $S_{\seg_u,v}$ as the set of all the distance values obtained after
  executing an action belonging to any maximal causal chain of~$\seg_u$ rooted
  at~process $v$ ($v \in V_u$)).
\end{definition}

Note that, from Observation~\ref{obs:5} and Lemma~\ref{lem:causalchain}, we have the following observation:

\begin{observation}\label{obs:finite}
The size of the set $S_{\seg_u,v}$ is bounded by a function of the
number of processes in $V_u$.
\end{observation}

\begin{lemma}\label{lem:Xnmax}
  Given a non-root process $u$ and a $u$-segment~$\seg_u$, if the size
  of~$S_{\seg_u,v}$ is bounded by~$\X$ for all~process $v \in V_u$,
  then the number of~$\compute$ executions done by $u$ in~$\seg_u$ is
  bounded by~$\X(\nmax-1)$.
\end{lemma}
\begin{proof}
  Except possibly the first, all $\compute$ executions done by a $u$
  in a $u$-segment~$\seg_u$ are done through the
  rule~$\mathrm{R_C}$. For all these, the variable $\dist_u$ is always
  decreasing. Therefore, all the values of $\dist_u$ obtained by the
  $\compute$ executions done by~$u$ are different. By definition of
  $S_{\seg_u,v}$ and by Lemma~\ref{lem:causalchain}, all these values
  belong to the set $\bigcup_{v \in V_u \setminus \{u\}}S_{\seg_u,v}$,
  which has size at most $\X(\nmax-1)$.
\end{proof}

By definition, each step contains at least one action, made by a
non-root process.  Let $u$ be any non-root process.  Assume that, in
any $u$-segment $\seg_u$, the size of~$S_{\seg_u,v}$ is bounded
by~$\X$ for all~process $v \in V_u$. So, the number of step of $u$ in
$\seg_u$ is bounded by~$\X(\nmax-1)+3$, by Lemma~\ref{lem:Xnmax} and
Corollary~\ref{cor:language}. Moreover, recall that each execution
contains at most $\nmax+1$ $u$-segments
(Observation~\ref{obs:seg}). So, $u$ executes in at most $\X\nmax^2
+3\nmax-\X+3$ steps. Finally, as $u$ is an arbitrary non-root process
and there are $n-1$ non-root processes, follows.

\begin{theorem}\label{theo:compl}
  If the size of~$S_{\seg_u,v}$ is bounded by~$\X$ for all~non-root
  process $u$,~for all $u$-segment $\seg_u$, and for all process $v$
  in $V_u$, then the total number of steps during any execution, is
  bounded by $(\X\nmax^2 +3\nmax-\X+3)(n-1)$.
\end{theorem}

Let $\maxi = \max_{\{u,v\} \in E} \omega(u,v)$. If all weights are
strictly positive integers, then the size of any~$S_{\seg_u,v}$, where
$u$ is a non-root process and $v \in V_u$, is bounded by~$\maxi\nmax$,
because~$S_{\seg_u,v} \subseteq [ds_{\seg_u,v} + 1, ds_{\seg_u,v} +
  \maxi(n_{cc}-1)]$, where $n_{cc} \leq \nmax+1$ is the number of
processes in $V_u$.  Hence, we deduce the following theorem from
Theorem~\ref{theo:compl}, Observation~\ref{obs:finite}, and
Corollary~\ref{coro:legterm}.

\begin{theorem}
  Algorithm~$\A$ is silent self-stabilizing under the distributed
  unfair daemon for the set $\mathcal{LC}_{\A}$ and, when all weights
  are strictly positive integers, its stabilization time in steps is
  at most~$[\maxi\nmax^3+(3-\maxi)\nmax+3](n-1)$, {\em i.e.},
  $O(\maxi\nmax^3n)$.
\end{theorem}

If all edges in~$G$ have the same weight~$\weight$, then the size
of~$S_{\seg_u,v}$, where $u$ is a non-root process and $v \in V_u$, is bounded by~$\nmax$. Indeed, in such a case, we
have $S_{\seg_u,v} \subset \{ ds_{\seg_u,v} + i.\weight ~|~ 1 \leq i \leq
n_{cc}-1\}$, where $n_{cc} \leq \nmax+1$ is the number of
processes in $V_u$. Hence, we obtain the following
corollary.

\begin{corollary}
  If all edges have the same weight, then the stabilization time in
  steps of Algorithm~$\A$ is at most $(\nmax^3 +2\nmax+3)(n-1)$, which
  is less than or equal to $n^4$ for all $n \geq 2$.
\end{corollary}

%% file: rounds.tex
\section{Round Complexity of Algorithm \A}\label{sect:rounds}

We now prove that every execution of Algorithm \A lasts at
most~$3\nmax+D$ rounds, where~$\nmax$ is the maximum number of
non-root processes in a connected component and $D$ is the
hop-diameter of the connected component containing~$r$, $V_r$.

The first lemma essentially claims that all processes that are in
illegal branches progressively switch to status~$EB$ within~$\nmax$
rounds, in order of increasing depth.

\begin{lemma}\label{lem:propagatingEB}
  Let $i \in \mathds{N}^*$.  Starting from the beginning of round~$i$,
  there does not exist any process both in state~$C$ and at depth less
  than~$i$ in an illegal branch.
\end{lemma}
\begin{proof}
  We prove this lemma by induction on~$i$. The base case~($i=1$) is
  vacuum, so we assume that the lemma holds for some integer~$i\geq
  1$.  From the beginning of round~$i$, no process can ever choose a
  parent which is at depth smaller than~$i$ in an illegal branch
  because those processes will never have status~$C$, by induction
  hypothesis.  Moreover, no process with status~$C$ can have its depth
  decreasing to $i$ or smaller by an action of one of its ancestors at
  depth smaller than $i$, because these processes have status~$EB$ and
  have at least one child not having status~$EF$. Thus, they cannot
  execute any rule. Therefore, no process can take state~$C$ at depth
  smaller or equal to~$i$ in an illegal branch.

  Consider any process~$u$ with status~$C$ at depth~$i$ in an illegal
  branch at the beginning of the round~$i$. $u \neq r$. Moreover, by
  induction hypothesis,~$u$ is an abnormal root, or the parent of~$u$
  is not in state~$C$ ({\em i.e.}, it is in the state~$EB$).  During
  round~$i$,~$u$ will execute rule~$\mathrm{R_{EB}}$
  or~$\mathrm{R_{C}}$ and thus either switch to state~$EB$ or join
  another branch at a depth greater than $i$.  This concludes the
  proof of the lemma.
\end{proof}

\begin{corollary}\label{cor:propagatingEB}
  After at most $\nmax$ rounds, the system is in a configuration from
  which no process in any illegal branch has status~$C$ forever.

  Moreover, once such a configuration is reached, each time a process
  executes a rule other than~$\mathrm{R_{EF}}$, this process is
  outside any illegal branch forever.
\end{corollary}

The next lemma essentially claims that, once no process in an illegal
branch has status~$C$ forever, processes in illegal branches
progressively switch to status~$EF$ within~at most $\nmax$ rounds, in
order of decreasing depth.

\begin{lemma}\label{lem:propagatingEF}
  Let $i \in \mathds{N}^*$. Starting from the beginning of
  round~$\nmax+i$, there does not exist any process at depth larger
  than~$\nmax-i+1$ in an illegal branch having the status~$EB$.
\end{lemma}
\begin{proof}
  We prove this lemma by induction on~$i$. The base case~($i=1$) is
  vacuum (by Observation~\ref{obs:size}), so we assume that the lemma
  holds for some integer~$i\geq 1$.  At the beginning of
  round~$\nmax+i$, no process at depth larger than~$\nmax-i+1$ has the
  status~$EB$ (by induction hypothesis) or status~$C$ (by
  Corollary~\ref{cor:propagatingEB}). Therefore, processes with
  status~$EB$ at depth~$\nmax-i+1$ in an illegal branch can execute
  the rule~$\mathrm{R_{EF}}$ at the beginning of round~$\nmax+i$.
  These processes will thus all execute within round~$\nmax+i$ (they
  cannot be neutralized as no children can connect to them). We
  conclude the proof by noticing that, from
  Corollary~\ref{cor:propagatingEB}, once round~$\nmax$ has
  terminated, any process in an illegal branch that executes either
  gets status~$EF$, or will be outside any illegal branch forever.
\end{proof}

The next lemma essentially claims that, after the propagation of
status $EF$ in illegal branches, the maximum length of illegal
branches progressively decreases until all illegal branches vanish.

\begin{lemma}\label{lem:propagatingI}
  Let $i \in \mathds{N}^*$.  Starting from the beginning of
  round~$2\nmax+i$, there does not exist any process at depth larger
  than~$\nmax-i+1$ in an illegal branch.
\end{lemma}
\begin{proof}
  We prove this lemma by induction on~$i$. The base case~($i=1$) is
  vacuum (by Observation~\ref{obs:size}), so we assume that the lemma
  holds for some integer~$i\geq 1$. By induction hypothesis, at the
  beginning of round~$2\nmax+i$, no process is at depth larger than or
  equal to~$\nmax-i+1$ in an illegal branch.  All processes in an
  illegal branch have the status~$EF$. So, at the beginning of
  round~$2\nmax+i$, any abnormal root satisfies the
  predicate~$\PReset$, they are enabled to execute
  either~$\mathrm{R_I}$, or~$\mathrm{R_R}$.
  So, all abnormal roots at the beginning of the round~$2\nmax+i$ are
  no more in an illegal branch at the end of this round: the maximal
  depth of the illegal branches has decreased, since by
  Corollary~\ref{cor:propagatingEB}, no process can join an illegal
  tree during the round~$2\nmax+i$.
\end{proof}

\begin{corollary}\label{cor:propagatingI}
After at most round~$3\nmax$, there are no illegal branches forever.
\end{corollary}

Note that in any connected component that does not contain the
root~$r$, there is no legal branch. Then, since the only way for a
process to be in no branch is to have status~$I$, we obtain the
following corollary.

\begin{corollary}\label{cor:convergeI}
  For any connected component~$H$ other than $V_r$, after at
  most~$3\nmax$ rounds, every process of~$H$ is in a legitimate state
  forever.
\end{corollary}

In the connected component~$V_r$, Algorithm~\A may need additional
rounds to propagate the correct distances to~$r$. 
In the next lemma, we use the notion of hop-distance to~$r$ 
defined below.

\begin{definition}[Hop-Distance and Hop-Diameter]
  A process~$u$ is said to be at \emph{hop-distance}~$k$ from~$v$ if
  the minimum number of edges in a shortest path from $u$ to $v$ is~$k$. 

  The \emph{hop-diameter} of a graph $G$ (resp. of a connected
  component $H$ of the graph $G$) is the maximum hop-distance between
  any two nodes of $G$ (resp. of $H$).
\end{definition}

\begin{lemma}\label{lem:lastDrounds} 
  Let~$i \in \mathds{N}$. In every execution of Algorithm~\A, starting
  from the beginning of round~$3\nmax+i$, every process at hop-distance
  at most~$i$ from~$r$ is in a legitimate state.
\end{lemma}
\begin{proof}
  We prove this lemma by induction on $i$.  First, by definition, the
  root $r$ is always in a legitimate state, so the base case~($i=0$)
  trivially holds. 
  Then, after at most $3\nmax$ rounds, every process either belongs to
  a legal branch or has status~$I$ (by
  Corollary~\ref{cor:propagatingI}), thus any non-isolated process~$v
  \in V_r$ always stores a distance $d$ such that $d \geq d(v,r)$, its
  actual weighted distance to~$r$.
  By induction hypothesis, every process at hop-distance at most~$i$
  from~$r$ has converged to a legitimate state within at most
  $3\nmax+i$ rounds. Therefore, at the beginning of
  round~$3\nmax+i+1$, every process~$v$ at hop-distance~$i+1$ from~$r$
  which is not in a legitimate state is enabled for executing
  rule~$\mathrm{R_C}$.  Thus, at the end of round~$3\nmax+i+1$, every
  process at hop-distance at most~$i+1$ from~$r$ is in a legitimate
  state (such processes cannot be neutralized during this round).
  Also, these processes will never change their state since there are
  no processes that can make them closer to~$r$.
\end{proof}

Summarizing all the results of this section, we obtain the following
theorem.

\begin{theorem}\label{theo:main-fair}
  Every execution of Algorithm \A lasts at most~$3\nmax+D$ rounds,
  where~$\nmax$ is the maximum number of non-root processes in a
  connected component and $D$ is the hop-diameter of the connected
  component containing~$r$.
\end{theorem}

% Recall that under a {\em weakly fair} daemon, every continuously
% enabled process is eventually activated by the daemon. By definition,
% every round is finite, yet maybe unbounded, in terms of steps under
% such an assumption.  Hence, if every execution contains a finite
% number of rounds, then every execution is finite under the weakly fair
% daemon assumption.
% 
% Notice then that all proofs made in this section still hold if we
% assume that edge weights are strictly positive real numbers. Hence
% 
% \begin{observation}
%   If edge weights are strictly positive real numbers, then
%   Algorithm~$\A$ is silent self-stabilizing under the distributed
%   weakly fair daemon for the set $\mathcal{LC}_{\A}$ and its
%   stabilization time in rounds is still at most $3\nmax+D$,
%   where~$\nmax$ is the maximum number of non-root processes in a
%   connected component and $D$ is the hop-diameter of the connected
%   component containing~$r$.
% \end{observation}

%% file: ccl.tex
\section{Conclusion}\label{ccl}

In this paper, we have proposed a silent self-stabilizing
  algorithm for the DCDSPM problem. This algorithm is written in the
  composite atomicity model, assuming a distributed unfair daemon (the
  weakest scheduling assumption of the model).
Its stabilization time in rounds is at most~$3\nmax+D$, where $\nmax$
is the maximum number of non-root processes in a connected component
and $D$ is the hop-diameter of $V_r$.
Furthermore, if we additionally assume that edge weights are positive
integers, then it stabilizes in a polynomial number of steps: namely,
we exhibit a bound in $O(\maxi \nmax^3 n)$, where $\maxi$ is the
maximum weight of an edge and $n$ is the number of processes.
To obtain this stabilization time polynomial in steps, the key idea
was to freeze the growth of abnormal trees before removing them in a
top-down manner. This freezing mechanism is implemented as a
propagation of information with feedback in the tree. This technique
is general. In particular, it can be used in other spanning tree or
forest constructions.

The stabilization time is, by definition, evaluated from an arbitrary
initial configuration, and so is drastically impacted by worst case
scenarios. Now, in many cases, transient faults are sparse and their
effect may be superficial. For example, a topological change in a
network commonly consists of a single link failure. Some
specializations of self-stabilization, such as
superstabilization~\cite{DH97j}, self-stabilization with service
guarantee~\cite{JM14}, or gradual stabilization~\cite{ADDP16c} have
been proposed to target recovery from such favorable cases as a
performance issue.
Proposing silent algorithms for the DCDSPM problem implementing one
of these aforementioned stronger properties, while achieving polynomial step
complexity, is an interesting perspective of our work.